\documentclass[11pt]{article}

\usepackage[margin=3 cm]{geometry}
\usepackage{graphicx,amssymb,amsmath,amsthm, latexsym} 
\usepackage{amsfonts}
\usepackage{hyperref}
\usepackage{enumerate}
\usepackage{comment}
\usepackage{mathtools}
\usepackage{multicol}

\usepackage[left]{lineno}

\usepackage{pgf,tikz}
\usepackage{mathrsfs}
\usetikzlibrary{arrows}
\usetikzlibrary{backgrounds}

\newtheorem{thm}{Theorem}[section]

\newtheorem{lem}[thm]{Lemma}
\newtheorem{prop}[thm]{Proposition}
\newtheorem{cor}[thm]{Corollary}

\theoremstyle{definition}
\newtheorem{definition}{Definition}[section]

    \newtheoremstyle{TheoremNum}
        {\topsep}{\topsep}              
        {\itshape}                      
        {}                              
        {\bfseries}                     
        {.}                             
        { }                             
        {\thmname{#1}\thmnote{ \bfseries #3}}
    \theoremstyle{TheoremNum}
    \newtheorem{thmn}{Theorem}
		
		    \newtheoremstyle{TheoremNum}
        {\topsep}{\topsep}              
        {\itshape}                      
        {}                              
        {\bfseries}                     
        {.}                             
        { }                             
        {\thmname{#1}\thmnote{ \bfseries #3}}
    \theoremstyle{TheoremNum}
    \newtheorem{propn}{Proposition}

\newcommand{\RR}{\mathbb{R}}      
\newcommand{\NN}{\mathbb{N}}      
\newcommand{\PP}{\mathbb{P}}      

\newcommand{\sgn}{\operatorname{sgn}}
\newcommand{\Dis}{\operatorname{Dis}}

\title{Learning Time Dependent Choice}
\author{Zachary Chase\thanks{%
California Institute of Technology,
{\sl zchase@caltech.edu}
} \and Siddharth Prasad\thanks{%
California Institute of Technology,
{\sl sprasad@caltech.edu}
}}
\begin{document}

\maketitle

\begin{abstract}
We explore questions dealing with the learnability of models of choice over time. We present a large class of preference models defined by a structural criterion for which we are able to obtain an exponential improvement over previously known learning bounds for more general preference models. This in particular implies that the three most important discounted utility models of intertemporal choice -- exponential, hyperbolic, and quasi-hyperbolic discounting -- are learnable in the PAC setting with VC dimension that grows logarithmically in the number of time periods. We also examine these models in the framework of active learning. We find that the commonly studied stream-based setting is in general difficult to analyze for preference models, but we provide a redeeming situation in which the learner can indeed improve upon the guarantees provided by PAC learning. In contrast to the stream-based setting, we show that if the learner is given full power over the data he learns from -- in the form of learning via membership queries -- even very naive algorithms significantly outperform the guarantees provided by higher level active learning algorithms.
\end{abstract}

\section{Introduction}

We study the learnability of economic models of choice over time. Our setting is that of an analyst who first observes an agent's choices between plans that specify payoffs over time, and then attempts to learn the preference parameters guiding the choices. While such parameters are stylized -- in reality subjects are not likely to perform standardized computations according to private parameters before making decisions -- experiments have shown that they often provide accurate descriptions of how an agent behaves. By observing enough choice data, one can hope to learn the economic parameters that most closely describe the agent's preferences. Thus, learning theory provides an especially meaningful lens with which to view the theory of choice -- it allows us to answer questions regarding the volume of data required to faithfully predict future decisions made by an observed agent. The overarching goal of this paper is to identify structural criteria that yield strong learnability results for preferences over time under different restrictions placed on the learner/analyst. The criteria we present captures a large class of preference models that give the agent significant freedom in weighting decisions against time delays. In particular, it encompasses the most popular models of time dependent choice used by economists. 

The main economic application of our results is in understanding the learnability of models of intertemporal choice. Intertemporal choice is what governs an agent's decisions over several time periods. The most important models of intertemporal choice are discounted utility models, in which agents evaluate plans by discounting actions as they are delayed -- in analogy to how markets value the loss or gain of money over time. The first axiomatic treatment of discounting was by Koopmans in 1960 \cite{KOOPMANS}, in which he demonstrates that simple postulates for preferences over an infinite time horizon yield ``impatience." The three most commonly studied discounting models are \emph{exponential}, \emph{hyperbolic}, and \emph{quasi-hyperbolic}, and all have been studied by both economists and computer scientists (though less so by the latter) as well as researchers from various other fields. The importance of discounted utility in economics cannot be overstated -- it is the canonical framework used by economists to study choice over time. 

Problems of learning economic parameters have received recent attention from computer scientists; see, e.g., \cite{BALCAN, FEDER, BEIGMAN, KALAI, ZADIM}. Inspired by a general theme of demanding computational robustness from economic models (Echenique, Golovin, and Wierman provide a nice discussion of this topic in \cite{ECHENIQUE}), the tools of learning theory provide relevant and exciting perspectives from which to view economic models that have been around for several decades. In contrast to the usual goal of truthfully extracting the agent's parameters adopted by classical mechanism design, the learning problem aims to efficiently extract a truthful agent's parameters in the restricted message space of binary classification. Our paper contributes to the line of work that specifically studies models of choice using the perspectives of learning theory. This confluence of decision theory and learning theory was initiated by Basu and Echenique \cite{FEDER}, who consider the learning problem for models of choice under uncertainty. Our investigation in this paper is motivated by models of how agents make choices over time. We provide learnability results that are fine tuned to structural requirements on such models.

We now summarize our main contributions at a high level. Section 3 contains a more detailed exposition of our results.

\subsubsection*{Summary of results and techniques}
Our situation is one of an analyst trying to learn the parameters governing an agent's preferences over time. The two main learning themes we consider are (1) when the analyst has no control over the data he sees and (2) when the analyst has some control over the data he sees. The first theme is aptly captured by \emph{probably approximately correct (PAC) learning}. To analyze the second theme, we investigate two models of \emph{active learning}: stream-based selective sampling and membership queries.

In the first part, we study the PAC model, where the analyst is presented with pairs of alternatives and a label for each pair indicating the agent's preference between the alternatives. The data points are drawn according to some unknown distribution, and the analyst has no control over the data he is presented with. Our main result here is a structural criterion on preference models that allows for a drastic improvement over the PAC learning complexity bounds achieved in \cite{FEDER}. We stipulate that the agent weights time-delayed payoffs according to polynomials, which allows for considerable freedom in how payoffs are weighted. Under this requirement, we show that such classes of preference models admit an exponential improvement in sample complexity bounds over the more general preference models considered in \cite{FEDER}. This is achieved via a computation of the VC dimension (which quantifies the complexity of PAC learning). A simple application of our result shows that each of the discounted utility models are learnable, with sample complexity that grows logarithmically in the number of time periods $T$ over which decisions are being made. The computation of the VC dimension is due to a natural connection between pairs of choices and the signs of polynomials that arise from the choices.

In the second part, we consider active learning models, where the analyst is given a certain amount of control over the data that he uses to learn. The two active learning models we study are \emph{stream-based selective sampling} and learning via \emph{membership queries}. In the former, the analyst is given some control over what data he learns from: as in the PAC setting he is presented with points drawn from an unknown distribution, but now the analyst chooses whether or not to see the label representing the agent's choice for each point. In the latter, the analyst has complete control over the data he learns from: the analyst can at any time request the label for any point. The former model seems to have been commonly adopted in order to study the very general problem of concept learning, when there is no extra information about the structure of the concepts. We find that the \emph{disagreement methods} used to study the stream-based setting are in general difficult to analyze in the context of preference models -- requiring quantitative information about the underlying distribution from which points are drawn. However, we provide a redeeming situation (by examining a particular distribution) where we obtain an improvement over the PAC guarantees. Membership queries, on the other hand, allow us to heavily exploit the structure of the preference models we consider. We present a naive membership query-based algorithm that significantly outperforms the guarantees provided in the stream-based setting. Learning via membership queries, we conclude, seems to be the appropriate model to actively learn economic parameters. It allows the analyst to make use of the preference relations' structure, and also precisely captures the situation in which the analyst and agent are participating in a real time experiment. 

\subsubsection*{Related work}

Discounted utility models of intertemporal choice have been studied extensively not only by economists, but also by researchers from various other fields. We first briefly survey some of the relevant work pertaining to the exponential, quasi-hyperbolic, and hyperbolic discounting models and then survey existing work in the more general topic of learning economic parameters.

In the exponential discounting model, the agent evaluates his utilities based on a discount factor $\delta\in (0, 1)$, where a delay of $t$ time periods incurs an exponential discount in utility by $\delta^t$. Climate change policies are traditionally evaluated according to an exponential discounting model -- for example, the \emph{Stern review} on the economics of climate change deals with issues of how to choose an appropriate discount rate in evaluating such policies \cite{STERN}. Chambers and Echenique \cite{CHAMBERS} present results related to the problem of aggregating discount rates proposed by a group of experts facing disagreement. While it is the most commonly used discounting model due to its simplicity, the exponential discounting model has been criticized due to its inability to match empirical data recording actual human behavior. Quasi-hyperbolic and hyperbolic discounting aim to mend such issues. The quasi-hyperbolic discounting model is parametrized by $\beta, \delta\in (0, 1)$, where a delay of $t$ time periods incurs a discount in utility by $\beta\delta^t$, and was first introduced by Phelps and Pollack \cite{PHELPS} to study preferences over generations. They proposed that the constant $\beta$ discount factor represents how much a given generation $t$ is affected by the utilities of other people relative to their own -- and remark that $\beta=1$ represents ``perfect altruism," while $\beta < 1$ represents ``imperfect altruism." Kleinberg and Oren \cite{KLEINBERG} study agents with quasi-hyperbolic discounting and propose a graph-theoretic model to investigate phenomena such as procrastination and abandonment of long-range tasks. Hyperbolic discounting aims to capture the notion that people are more impatient in making short term decisions (today vs. tomorrow) than long term decisions (365 days from today vs. 366 days from today)\footnote{In particular note that exponential discounting does not capture this issue, i.e. it is \emph{dynamically consistent}, in that preferences do not change according to shifts in time.}, and is modeled via a discount of $(1+t\alpha)^{-1}$ at time $t$. Researchers in fields such as psychology and neuroscience \cite{BERNS, KABLE} have adopted the hyperbolic discounting model to study, for example, issues of self control and anticipation in humans and animals, and have compared the predictions by the different discounted utility models to neurobiological data obtained via MRI scans. Chabris et al. \cite{CHABRIS} give an exposition of the discounted utility models of intertemporal choice and survey sociological research that examines empirical data pertaining to how discount rates are affected by factors like age, drug use, gambling, etc.

The study of economic models has witnessed a recent influx of work from computer scientists dealing with questions of robustness under various notions of complexity (learning complexity, computational complexity, communication complexity, etc.). Kalai \cite{KALAI} in 2001 studied the learnability of choice functions, where the observed choices are in the form of a given set of alternatives along with the most preferred alternative from the set. Beigman and Vohra \cite{BEIGMAN}, Zadimoghaddam and Roth \cite{ZADIM}, and Balcan et al. \cite{BALCAN} investigate the problem of learning utility functions in the context of an expected utility maximizing agent in a demand environment. Most recently (and most related to our work), Basu and Echenique \cite{FEDER} study the learnability of preference models of choice under uncertainty, in which an agent is uncertain about states of a lottery and is made to choose between acts that encode utilities over each state. Here, the different models of choice under uncertainty arise from different ways of representing the subjective probability held by an agent. They are also the first to study learnability in the decision-theoretic setting where choice is modeled by preference relations rather than by expected utility maximizing behavior in a demand setting. However, it does not appear that the learnability of models of intertemporal choice has been previously studied.

\section{Model and Preliminaries}

We now formally set up the discounted utility models of intertemporal choice and state the standard definitions from learning theory in the context of preference relations. Much of the following material regarding learning and preference relations is taken from \cite{FEDER} since we require a similar list of definitions and setup. First, we sketch our high level model. 

Let $X$ be a Euclidean space equipped with a Borel $\sigma$-algebra. A \emph{preference relation} on $X$ is a binary relation $\succsim\subseteq X\times X$ such that $\succsim$ is measurable with respect to the product $\sigma$-algebra on $X\times X$. A \emph{model} $\mathcal{P}$ of preference relations is a collection of preference relations.

An agent makes choices from pairs of alternatives $(x^i, y^i)_{i=1}^n$ that are drawn according to some unknown distribution on $X\times X$. The choices are presented as labels $(a_i)_{i=1}^n$ where $a_i = 1$ if the agent chooses $x^i$ and $a_i = 0$ if the agent chooses $y^i$. A \emph{dataset} is any finite sequence of pairs of plans and their labels $(((x^1, y^1), a_1), \ldots, ((x^n, y^n), a_n))$. An analyst observes a dataset, and attempts to guess the preference relation governing the agent's choices. A \emph{learning rule} is any map $\sigma$ from datasets to preference relations. The output of the learning rule is the analyst's hypothesis as to what the agent's true preference relation is, having seen some finite dataset.

\subsection{Learnability}

The two notions of learnability we consider are the PAC model and the active model. We now state the standard definitions of PAC and active learning in the context of preference relations. Most of the following definitions for the PAC setting are taken from \cite{FEDER} since the setup involving preference relations is identical. These definitions of course apply to the more general setting of concept learning (for example, see \cite{BLUMER}).

A collection $\mathcal{P}$ of preference relations is (PAC) learnable if there is a learning rule $\sigma$ such that for every $0<\varepsilon, \delta < 1$, there is $s(\varepsilon,\delta)\in\NN$ such that for every $n\ge s(\varepsilon, \delta)$, $\succsim\in\mathcal{P}$, and $\mu\in\Delta(X\times X)$, $$\mu^n(\{((x^1, y^1),\ldots, (x^n, y^n)) : \mu(\succsim^{*} \triangle \succsim) > \varepsilon \}) < \delta,$$ where $$\succsim^{*} = \sigma(\{((x^1, y^1), I_{x^1\succsim y^1}), \ldots, ((x^n, y^n), I_{x^n\succsim y^n})\})$$ is the hypothesis preference relation produced by the learning rule\footnote{$\mu^n$ denotes the product measure induced by $\mu$ on $(X\times X)^n$.}. The quantity $s(\varepsilon, \delta)$ is called the \emph{sample complexity} of the learning rule $\sigma$. 

The complexity of learning is commonly quantified by the Vapnik-Chervonenkis (VC) dimension, which we now define. A set of points $\{(x^1, y^1),\ldots, (x^n, y^n)\}$ from $X\times X$ is \emph{shattered} by a model of preferences $\mathcal{P}$ if for every vector of labels $(a_1,\ldots, a_n)\in\{0, 1\}^n$, there is a preference relation $\succsim\in\mathcal{P}$ that realizes the labelling, i.e. for $i=1,\ldots, n$ we have that $x^i\succsim y^i$ if and only if $a_i = 1$. In this case, $\mathcal{P}$ is said to \emph{rationalize} the dataset $\{((x^1, y^1), a_1), \ldots, ((x^n, y^n), a_n)\}$. The VC dimension of $\mathcal{P}$, denoted by $VC(\mathcal{P})$, is the largest integer $n$ such that there exist $n$ points that are shattered by $\mathcal{P}$.

Blumer et al. \cite{BLUMER} in 1989 proved that learnability is equivalent to having a finite VC dimension\footnote{This result requires $\mathcal{P}$ to satisfy a certain measurability requirement. We note in Section 4 that the models of choice we consider all satisfy said requirement.}.

\begin{thm}\label{VCEquivalence} A model of preferences $\mathcal{P}$ is learnable if and only if $VC(\mathcal{P}) < \infty$. \end{thm}

The VC dimension (denoted by $d$ for the remainder of this subsection) also plays a role in the sample complexity of learning a model of preferences. In the same paper, Blumer et al. \cite{BLUMER} show that any algorithm that outputs a hypothesis consistent with the data seen is a valid learning rule requiring sample complexity $$s(\varepsilon, \delta) = O\left(\frac{1}{\varepsilon}\left(d\log\frac{1}{\varepsilon} + \log\frac{1}{\delta}\right)\right).$$

In 2016, Hanneke \cite{HANNEKE2} showed that these bounds (after a small improvement) are tight: the optimal sample complexity of PAC learning is $$s(\varepsilon, \delta) = \Theta\left(\frac{1}{\varepsilon}\left(d+\log\frac{1}{\delta}\right)\right).$$

The other learning model we consider is the active learning framework, where the analyst has some control over the data from which he learns. In stream-based selective sampling, points drawn according to an unknown distribution are presented to the analyst as before, but without the labels. The analyst can choose whether or not to query the label of a given point, and the complexity of the learning rule is measured by \emph{label complexity}, i.e. the number of labels requested by the analyst. Disagreement based active learning refers to the paradigm in which the learner only requests labels on points that significantly reduce the hypothesis space. The disagreement of a preference model with respect to the underlying distribution is quantified through the \emph{disagreement coefficient} $\theta$, which is defined in Section 5. A finite disagreement coefficient implies (for the underlying distribution) an exponential improvement in label complexity over the sample complexity of PAC learning. For example, the CAL algorithm \cite{COHN, DASGUPTA, HANNEKE}, a simple disagreement based learning algorithm, yields a label complexity of $$\ell_{CAL}(\varepsilon, \delta) = O\left(\theta \log\frac{1}{\varepsilon}\left(d\log\theta + \log\frac{\log(1/\varepsilon)}{\delta}\right)\right).$$ 

In the membership queries model, the analyst is allowed to request the label for any point at any time. There appears to be a dearth of literature/results pertaining to the complexity of membership query algorithms for learning when the hypothesis space is infinite. One explanation for this is that improvements to the ``passive" disagreement based methods used in the stream-based setting would need specific information about the problem domain: disagreement based methods are designed to work on a very general class of concept learning problems without assuming anything about the learning space. In our case, we have specific details about how the preference relations take shape. Thus, the membership query model turns out to be an interesting and useful perspective to use in the study of learning preference models.

For a more detailed survey of active learning, see \cite{DASGUPTA}.

\subsection{Discounted utility}

We now present the definitions for the discounted utility models of intertemporal choice. An agent chooses between \emph{plans} or vectors in $X = \RR^T$ that encode payoffs over $T$ time periods. A preference relation over plans is a binary relation $\succsim\subseteq\RR^T\times\RR^T$.

The most important model of intertemporal choice is the \emph{discounted utility model}, in which the agent's payoffs $x_t$ for having chosen a plan $x\in\RR^T$ are reduced, or discounted, as $t$ increases from $1$ to $T$. In its most general form, we can characterize the preference relations that follow time discounting as follows:

\begin{definition}[Discounted utility model] The class of preference relations $\mathcal{P}_{\mathcal{D}}$ that satisfy the \emph{discounted utility model} are those $\succsim$ such that there exists a decreasing map $D:\{1,\ldots, T\}\to (0, 1)$ where $$x\succsim y \,\,\text{ if and only if }\,\, \sum_{t=1}^T D(t)x_t\ge\sum_{t=1}^T D(t) y_t.$$ \end{definition}

We use the following notation for the preference models arising from the three most commonly studied discounting functions $D$:

\begin{itemize}
\item $\mathcal{P}_{\mathcal{D}}$ denotes the set of preferences that satisfy the discounted utility model.
\item $\mathcal{P}_{\mathcal{ED}}$ denotes the set of preferences that satisfy the discounted utility model with \emph{exponential discounting}: $D(t) = \delta^t$ for $\delta\in (0, 1)$.
\item $\mathcal{P}_{\mathcal{HD}}$ denotes the set of preferences that satisfy the discounted utility model with \emph{hyperbolic discounting}: $D(t) = \frac{1}{1+t\alpha}$ for $\alpha > 0$.
\item $\mathcal{P}_{\mathcal{QHD}}$ denotes the set of preferences that satisfy the discounted utility model with \emph{quasi-hyperbolic discounting}: $D(t) = 1$ if $t = 1$, $D(t) = \beta\cdot\delta^{t-1}$ if $t > 1$ for $\beta,\delta\in (0, 1)$.
\end{itemize}

For a more thorough exposition on the various discounted utility models of intertemporal choice, see \cite{CHABRIS}.

\section{Main Results}

In this section we provide a formal discussion and interpretation of our results, which is split into two themes: the first dealing with an analyst who has no control over the learning data, the second dealing with an analyst who has some control over the learning data.

\subsection*{A powerless analyst}

The first part of our paper investigates the situation of an analyst trying to learn the preference relation by which an agent makes choices, but has no control over what choices he gets to observe, and is agnostic to the process by which they are drawn. We thus adopt the PAC learning model.

The agent chooses between plans that encode payoffs over $T$ periods of time and evaluates the total payoff of a plan vector $x\in\RR^T$ according to private weights $w_1,\ldots, w_T$ that he multiplicatively applies to each state: $\text{payoff}(x) = \sum_{t=1}^T w_tx_t.$ This defines a model of preference relations, which we denote by $\mathcal{P}_{\mathcal{W}}$, where for any $\succsim\in\mathcal{P}_{\mathcal{W}}$, there exists a vector of weights $w = (w_1,\ldots, w_T)\in\RR^T$ such that $$x\succsim y \, \, \text{ if and only if } w.x\ge w.y.$$ In \cite{FEDER}, it is shown that $T-1 \le VC(\mathcal{P}_{\mathcal{W}}) \le T+1$. In the context of choice over time, however, this model is extremely general and does not capture any of the intuitive notions of how an agent values payoffs when they are delayed\footnote{In \cite{FEDER} the complete control over weights is used to model choice under uncertainty, which calls for such generality since the agent's beliefs/weights are given by an element of the probability simplex on $\RR^T$.}. For example, the discounted utility models of intertemporal choice require the weights to be of a particular functional form. Moreover, when there is no structure to the discount function we cannot improve the bounds on $\mathcal{P}_{\mathcal{W}}$:

\begin{prop}\label{GeneralDiscounting} $T-1\le VC(\mathcal{P}_{\mathcal{D}})\le T+1$. \end{prop}

This leads us to the motivating question of the first part of the paper: what structural conditions can we impose on the weights $w_1,\ldots, w_T$ such that this bound can be improved?

We investigate the situation where the agent computes his weights by evaluating polynomials at a private parameter $\delta$. Specifically, let $Q_1,\ldots, Q_T$ be polynomials of degree at most $d$, and suppose the agent evaluates total payoff of a plan vector $x\in\RR^T$ by $\text{payoff}(x) = \sum_{t=1}^T Q_t(\delta)x_t$. Consequently, let $\mathcal{P}_{\mathcal{PW}}$ be the model of preference relations parametrized by $\delta$ such that $$x\succsim y\,\,\text{ if and only if }\sum_{t=1}^{T}Q_t(\delta)x_t\ge\sum_{t=1}^T Q_t(\delta)y_t.$$

This class of preference models allows us to approximate preference relations where the weights are given by any real valued functions -- we choose $Q_1,\ldots, Q_T$ to be the appropriate Taylor polynomials. Moreover, existing models of intertemporal choice fit this characterization -- for example $\mathcal{P}_{\mathcal{ED}}$ and $\mathcal{P}_{\mathcal{HD}}$. 

We additionally consider a slightly larger class of preference models where the agent has a private parameter $\beta$ (in addition to $\delta$) that in evaluating total payoff of a plan vector $x\in\RR^T$ allows the agent to modify the constant term $\sum_{t=1}^{T}Q_t(0)x_t$ of the polynomial $\sum_{t=1}^{T} Q_t(\delta)x_t$. This model aims to more generally capture the effects of the $\beta$ parameter in quasi-hyperbolic discounting. For polynomials $Q_1,\ldots, Q_T$ of degree at most $d$, let $\mathcal{P}_{\mathcal{BPW}}$ be the model of preference relations parametrized by $\beta$ and $\delta$ such that $x\succsim y$ if and only if $$\left(\frac{1}{\beta}-1\right)\sum_{t=1}^T Q_t(0)x_t+\sum_{t=1}^T Q_t(\delta)x_t \ge \left(\frac{1}{\beta}-1\right)\sum_{t=1}^T Q_t(0)y_t+\sum_{t=1}^T Q_t(\delta)y_t.\footnote{We use a factor of $(1/\beta-1)$ since it yields a clean description of quasi-hyperbolic discounting.}$$

Our main results show that with this additional structure on the preference model, we can achieve an exponential improvement in the bounds for the VC dimension of $\mathcal{P}_{\mathcal{W}}$ obtained in \cite{FEDER}\footnote{It is important to note that the classes $\mathcal{P}_{\mathcal{PW}}$ and $\mathcal{P}_{\mathcal{BPW}}$ are defined for a given $Q_1,\ldots, Q_T$. That is, the analyst knows $Q_1,\ldots, Q_T$, and is trying to learn the parameters $\beta$ and $\delta$. If the $Q_1,\ldots, Q_T$ are private information only available to the agent, we are in no better shape than in the case of $\mathcal{P}_{\mathcal{W}}$.}.

\begin{thm}\label{VC_upper} For every $\varepsilon > 0$, there exists a $d_{\varepsilon}$ such that for every $d\ge d_{\varepsilon}$ we have $VC(\mathcal{P}_{\mathcal{PW}}), VC(\mathcal{P}_{\mathcal{BPW}})\le (1+\varepsilon)\log d$ for any $T$ and any $T$ polynomials $Q_1,\ldots, Q_T$ of degree at most $d$.
\end{thm}

Note that when $Q_1,\ldots, Q_T$ have degree at most polynomial in $T$, we obtain an exponential improvement over the linear growth of $VC(\mathcal{P}_{\mathcal{W}})$. We show that in this case, we get a tight (asymptotic) bound of $\log T$:

\begin{thm}\label{VC_lower} Let $Q_1,\ldots, Q_T$ be polynomials in $\delta$ of degree at most $T-1$ that span the space of polynomials in $\delta$ of degree at most $T-1$. Then $VC(\mathcal{P}_{\mathcal{PW}}), VC(\mathcal{P}_{\mathcal{BPW}})\ge\log(T-1)$. \end{thm}

An interesting feature of Theorems \ref{VC_upper} and \ref{VC_lower} is that for fixed $Q_1,\ldots, Q_T$ with degrees at most $T-1$, $VC(\mathcal{P}_{\mathcal{PW}})$ and $VC(\mathcal{P}_{\mathcal{BPW}})$ satisfy the same asymptotic bounds, so giving the agent an extra parameter that allows control over the constant term of the polynomial $\sum_{t=1}^T Q_t(\delta)x_t$ does not introduce a significant amount of richness to the model.

Applying Theorems \ref{VC_upper} and \ref{VC_lower} to the discounted utility models, we have:

\begin{cor}\label{Discount_bounds}$VC(\mathcal{P}_{\mathcal{ED}}), VC(\mathcal{P}_{\mathcal{HD}}), VC(\mathcal{P}_{\mathcal{QHD}})\sim\log(T-1)$ \end{cor}

Thus, $\mathcal{P}_{\mathcal{D}}$, $\mathcal{P}_{\mathcal{ED}}$, $\mathcal{P}_{\mathcal{HD}}$, and $\mathcal{P}_{\mathcal{QHD}}$ are all learnable. $\mathcal{P}_{\mathcal{D}}$ requires a minimum sample size that grows linearly with $T$, while $\mathcal{P}_{\mathcal{ED}}$, $\mathcal{P}_{\mathcal{HD}}$, and $\mathcal{P}_{\mathcal{QHD}}$ require a minimum sample size that grows logarithmically in $T$.

The main technique in proving Theorems \ref{VC_upper} and \ref{VC_lower} is interpreting the shattering criteria as a statement about the sign combinations achieved by a collection of polynomials. The upper bound on the VC dimension follows from an upper bound on the number of sign combinations a collection of polynomials can achieve. In demonstrating the lower bound on the VC dimension, we construct a set of points that is shattered by finding polynomials achieving all possible sign combinations -- chosen according to a Hamiltonian path in the $\log(T-1)$-dimensional hypercube.

\subsection*{A powerful analyst}

Our other results concern the active learning framework, which broadly deals with situations in which the analyst has some control over the choices he observes and learns from. The two models we consider are \emph{stream-based selective sampling} and learning via \emph{membership queries}. A large body of active learning research is devoted to the stream-based model, specifically focusing on disagreement based algorithms -- a class of learning algorithms that instructs the analyst only to request labels on points he sees that reduce the hypothesis space significantly. In the most general setting of concept learning, this is a useful framework since the error guarantees can be described using the same setup as the PAC model. Moreover without additional information about the problem domain, it is unclear how to devise efficient algorithms that are more specific in instructing the analyst on what questions to ask.

We find that the stream-based model is in general difficult to analyze for the preference relations we work with. This difficulty seems to arise from the apparent need to quantify disagreement in order to explicitly write down learning guarantees. Though in most general situations it is unclear how to quantify disagreement for our preference relations, we present a redeeming situation for which we are able to provide a precise analysis of the learning guarantees for $\mathcal{P}_{\mathcal{ED}}$. Here, the analyst can learn $\mathcal{P}_{\mathcal{ED}}$ with an exponential improvement in label complexity over the guarantees provided by the PAC model. This is achieved via a computation of the disagreement coefficient (defined in Section 5) of $\mathcal{P}_{\mathcal{ED}}$ for a specific distribution. 

\begin{thm}\label{finitemu}
There exists a distribution $\mu$ on $\RR^T\times\RR^T$ for which the disagreement coefficient of $\mathcal{P}_{\mathcal{ED}}$ is $\theta = 2$. Thus, for this distribution, $$\ell_{CAL}(\varepsilon) = \widetilde{O}\left(\log T\log\frac{1}{\varepsilon}\right),$$ where the $\widetilde{O}$ notation suppresses terms that are logarithmic in $\log T$ and $\log 1/\varepsilon$.
\end{thm}

The measure $\mu$ we construct is induced by the product Lebesgue measure on $(0, 1)^{T-1}$, and allows us to precisely translate statements about disagreement into statements about the roots of polynomials arising from a given choice. Once we have defined $\mu$, the calculation of $\theta$ follows from basic probability arguments.

Now, in our case the analyst has structural information regarding the preference relation of the agent he is questioning. We find that allowing the analyst full control over the membership queries he makes yields a learning algorithm that, despite its simplicity, takes advantage of this extra structure and yields a significant improvement in complexity over the stream-based setting. Additionally, the membership queries model naturally describes an experimental environment in which the analyst is able to ask the agent questions in real time.

We show that when the preference model satisfies some relatively benign structural requirements, even very naive algorithms outperform the guarantees provided by CAL in the stream-based setting. The example algorithm we give, relying on a simple binary search, has a query complexity of $O(\log 1/\varepsilon)$, which gets rid of the $\log T$ dependence in Theorem \ref{finitemu}.

The class of preference models is defined as follows: let $g_1,\ldots, g_T:\RR\to\RR$ be a collection of functions satisfying the properties listed in Section \ref{sec:memqueries} and consider the model of preference relations $\mathcal{P}$ parametrized by $\delta$ where $x\succsim y$ if and only if $\sum_{t=1}^T g_t(\delta)x_t\ge\sum_{t=1}^T g_t(\delta)y_t$\footnote{As before, the $g_1,\ldots, g_T$ are known to the analyst.}. We have

\begin{prop}\label{MembershipQueries}There exists an algorithm that takes as input $\varepsilon > 0$ and using $O(\log 1/\varepsilon)$ membership queries outputs $\delta^h$ such that $|\delta - \delta^h|\le\varepsilon$, where $\delta$ parametrizes the target preference relation in $\mathcal{P}$. \end{prop}

The remainder of the paper is devoted to proving the results discussed in this section.

\section{PAC Learning}

In this section we prove Theorems \ref{VC_upper} and \ref{VC_lower}. We first note a preliminary upper bound due to Basu and Echenique \cite{FEDER}. Let $\mathcal{P}_{\mathcal{I}}$ be the set of preference relations that satisfy the following axioms:

\begin{description}
\item[Order:] For all $x, y$ either $x\succsim y$ or $y\succsim x$ (completeness). For all $x, y, z$, if $x\succsim y$ and $y\succsim z$, then $x\succsim z$ (transitivity). 
\item[Independence:] For all $x, y, z$ and for any $\lambda\in (0, 1)$, $x\succsim y$ if and only if $\lambda x + (1-\lambda)z\succsim\lambda y + (1-\lambda)z$.
\end{description}

The class $\mathcal{P}_{\mathcal{I}}$ satisfies the property that for any $\succsim\in\mathcal{P}_{\mathcal{I}}$, there are finitely many vectors $q_1,\ldots, q_K$, with $K\le T$, such that $x\succsim y$ if and only if $(q_k.x)_{k=1}^K \ge_{L} (q_k.y)_{k=1}^K$, where $\ge_{L}$ denotes the lexicographic order \cite{BLUME}. Then, $\mathcal{P}_{\mathcal{PW}}, \mathcal{P}_{\mathcal{BPW}}\subset\mathcal{P}_{\mathcal{W}}\subset\mathcal{P}_{\mathcal{I}}$, since the aforementioned characterization is satisfied with $K=1$ and $q_1 = (w_1,\ldots, w_T)$.

This has two main consequences. First, $\mathcal{P}_{\mathcal{PW}}$ and $\mathcal{P}_{\mathcal{BPW}}$ (and thus all the discounted utility models) satisfy the measurability requirement discussed in Lemma 4 of \cite{FEDER} for the equivalence result of Theorem \ref{VCEquivalence} to hold. Second, the VC dimensions of $\mathcal{P}_{\mathcal{PW}}$ and $\mathcal{P}_{\mathcal{BPW}}$ are all bounded above by $T+1$ (and in particular $T-1\le VC(\mathcal{P}_{\mathcal{W}}) \le T+1$). This follows due to Theorem 3.1 of \cite{FEDER}, in which an argument similar to that required to compute the VC dimension of the class of half-spaces is used to show that $VC(\mathcal{P}_{\mathcal{I}}) = T+1$. In all cases excluding the most general model of discounted utility, we are able to bring this down to $\log(T-1)$ (which we then show is tight by demonstrating the corresponding lower bound).

We begin by demonstrating that even in the discounted utility setting, without any structure we cannot do better than the learning bounds obtained for $\mathcal{P}_{\mathcal{I}}$.

\begin{propn}[\ref{GeneralDiscounting}] $T-1\le VC(\mathcal{P}_{\mathcal{D}})\le T+1$. \end{propn}
\begin{proof}
That $VC(\mathcal{P}_{\mathcal{D}})\le T+1$ follows from Theorem 3.1 of \cite{FEDER}, since $\mathcal{P}_{\mathcal{D}}\subset\mathcal{P}_{\mathcal{I}}$. 

Here is a simple construction that shows $VC(\mathcal{P}_{\mathcal{D}})\ge T - 1$. Fix an $\varepsilon > 0$. Let $e_1,\ldots, e_T$ be the standard unit vectors in $\RR^T$, and consider the set of points $\{(x^1, y^1),\ldots, (x^{T-1}, y^{T-1})\}$, where $x^i = (1-\varepsilon)e_i$ and $y^i = e_{i+1}$.

This set is shattered by $\mathcal{P}_{\mathcal{D}}$: for any $(a_i)_{i=1}^{T-1}$, choose $D(1)$ arbitrarily from $(0, 1)$, and if $D(i)$ has been defined, inductively define $D(i+1)$ such that $D(i+1)\le D(i)(1-\varepsilon)$ if $a_i = 1$ and $D(i) > D(i+1) > (1-\varepsilon)D(i)$ if $a_i = 0$.
\end{proof}

We now prove Theorems \ref{VC_upper} and \ref{VC_lower}, which are restated below for convenience. 

\begin{thmn}[\ref{VC_upper}] For every $\varepsilon > 0$, there exists a $d_{\varepsilon}$ such that for every $d\ge d_{\varepsilon}$ we have $VC(\mathcal{P}_{\mathcal{PW}}), VC(\mathcal{P}_{\mathcal{BPW}})\le (1+\varepsilon)\log d$ for any $T$ and any $T$ polynomials $Q_1,\ldots, Q_T$ of degree at most $d$.\end{thmn}

\begin{proof}
It suffices to establish the bound for $\mathcal{P}_{\mathcal{PBW}}$.

Let $(z^1,\ldots, z^n)$ be a set of points in $\RR^T\times\RR^T$, $z^i = (x^i, y^i)$. For each $z^i = (x^i, y^i)$, define the plan $f^i := x^i - y^i$. Then, note that $(z^1,\ldots, z^n)$ is shattered by $\mathcal{P}_{\mathcal{PBW}}$ if and only if $((f^1, 0),\ldots, (f^n, 0))$ is shattered by $\mathcal{P}_{\mathcal{PBW}}$. Hence, we may (and do) restrict attention to datasets of the form $((f^1, 0),\ldots, (f^n, 0))$.

We have that $((f^1, 0),\ldots, (f^n, 0))$ is shattered by $\mathcal{P}_{\mathcal{PBW}}$ if and only if for all vectors $(a_1,\ldots, a_n)\in\{0, 1\}^n$, there exists a $\delta$ and $\beta$ (which determines the preference relation) such that $$(Q_T(\delta)f^i_T+\cdots+ Q_1(\delta)f^i_1) + \left(\frac{1}{\beta}-1\right)(Q_T(0)f^i_T+\cdots+ Q_1(0)f^i_1)\ge 0 \text{ whenever } a_i = 1,$$ and $$(Q_T(\delta)f^i_T+\cdots+ Q_1(\delta)f^i_1) + \left(\frac{1}{\beta}-1\right)(Q_T(0)f^i_T+\cdots+ Q_1(0)f^i_1) < 0 \text{ whenever } a_i = 0.$$ 

We first show that for all $\varepsilon > 0$, for sufficiently large $d$ we have $VC(\mathcal{P}_{\mathcal{PBW}})\le (1+\varepsilon)\log d$. Note that if the $n$ points $((f^1, 0), \ldots, (f^n, 0))$ can be shattered, there are polynomials $P_1,\ldots, P_n$ in $\delta$ (where $P_i$ is the polynomial $Q_1(\delta)f^i_1+\cdots+Q_{T}(\delta)f^i_T$), each of degree at most $d$, such that for every labeling $(a_1,\ldots, a_n)\in \{0, 1\}^n$, there exists a $\delta$ and $\beta$ such that $$(\sgn(P_1(\delta) + (1/\beta-1)P_1(0)),\ldots, \sgn(P_n(\delta) + (1/\beta-1)P_n(0))) = (a_1,\ldots, a_n) \footnote{For notational convenience, let $\sgn(x)$ be $1$ if $x\ge 0$ and $0$ otherwise.}.$$

First, for any $n$ polynomials $P_1,\ldots, P_n$ of degree at most $d$, we give an upper bound on the number of possible values $(\sgn(P_1(\delta)),\ldots, \sgn(P_n(\delta)))$ can realize. Each polynomial has at most $d$ real roots, so together $P_1,\ldots, P_n$ have at most $nd$ distinct real roots. Since sign changes can only occur at the roots, there are at most $nd+1$ possible values of $\{0, 1\}^n$ that $(\sgn(P_1(\delta)),\ldots, \sgn(P_n(\delta)))$ can realize. 

Now, for a fixed $\delta$, varying $\beta$ shifts the collection of polynomials $$P_1(\delta) + (1/\beta-1)P_1(0),\ldots, P_n(\delta) + (1/\beta-1)P_n(0)$$ vertically, which in the worst case induces sign changes in all entries. We thus get at most an additional $n$ new sign combinations for every sign combination realized by $(\sgn(P_1(\delta)),\ldots, \sgn(P_n(\delta)))$. Hence, there are at most $nd + 1 + n(nd + 1) = (n^2+n)d + n + 1$ possible values of $\{0, 1\}^n$ that $$(\sgn(P_1(\delta) + (1/\beta-1)P_1(0)),\ldots, \sgn(P_n(\delta) + (1/\beta-1)P_n(0)))$$ can realize.

In order for all $2^n$ elements of $\{0, 1\}^n$ to be realized, it must be that $$(n^2+n)d + n+1\ge 2^n.$$ If $n > (1+\varepsilon)\log d$, then for large enough $d$ this inequality does not hold, and so any set of $n$ points cannot be shattered. Thus, for all $\varepsilon > 0$, $n\le (1+\varepsilon)\log d$ for large enough $d$, i.e. $VC(\mathcal{P})\le (1+\varepsilon)\log d$.\end{proof}

We now establish the corresponding lower bound when the polynomials $Q_1,\ldots, Q_T$ span the space of polynomials of degree at most $T-1$.

\begin{thmn}[\ref{VC_lower}] Let $Q_1,\ldots, Q_T$ be polynomials in $\delta$ of degree at most $T-1$ that span the space of polynomials in $\delta$ of degree at most $T-1$. Then $VC(\mathcal{P}_{\mathcal{PW}}), VC(\mathcal{P}_{\mathcal{BPW}})\ge\log(T-1)$. \end{thmn}

\begin{proof}
It suffices to establish the bound for $\mathcal{P}_{\mathcal{PW}}$.

Consider the graph on $\{0,1\}^n$ where two vertices are connected by an edge if they differ in exactly one location. Fix a Hamiltonian path $v_1, v_2, \ldots, v_{2^n}$ in this graph (the existence of which is well known). Let $b_{1, 2},\ldots, b_{2^{n}-1, 2^n}$ be the sequence where $b_{i, i+1}$ is the index of the location at which $v_i$ and $v_{i+1}$ differ. Note that if $n = \log (T-1)$, the graph has $T-1$ vertices, so each index in $\{1,\ldots, n\}$ can appear in the sequence $(b_{i, i+1})$ at most $T - 1$ times. 

Now, let $r_1 < r_2 <\cdots < r_{2^n}$ be any points in $(0, 1)$. Define $n$ polynomials $P_1,\ldots, P_n$ by $P_k(\delta) = \prod_{b_{i, i+1} = k}(\delta-r_i),$ so the roots of $P_k$ are precisely the $r_i$'s that correspond to a flip in the entry at the $k$th position of a vertex in the path. Then, $(\sgn(P_1(\delta)),\ldots, \sgn(P_n(\delta)))$ realizes every element of $\{0, 1\}^n$.

Since $Q_1,\ldots, Q_T$ span the space of polynomials of degree at most $T-1$, for each $P_i$ we can find $f^i_1,\ldots, f^i_T$ such that $$P_i(\delta) = Q_1(\delta)f^i_1+\cdots + Q_T(\delta) f^i_T,$$ which gives us a collection of $\log(T-1)$ points that is shattered. Hence $\log(T-1)\le VC(\mathcal{P})$.
\end{proof}

It is readily seen that $\mathcal{P}_{\mathcal{ED}}$ and $\mathcal{P}_{\mathcal{QHD}}$ satisfy the conditions of Theorems \ref{VC_upper} and \ref{VC_lower}\footnote{$\mathcal{P}_{\mathcal{ED}}$ is given by $\mathcal{P}_{\mathcal{PW}}$ with $Q_t(\delta) = \delta^{t-1}$, and $\mathcal{P}_{\mathcal{QHD}}$ is given by $\mathcal{P}_{\mathcal{BPW}}$ with $Q_t(\delta) = \delta^{t-1}$.}. We give a quick argument verifying that $\mathcal{P}_{\mathcal{HD}}$ does as well: For $1\le t\le T$, let $Q_t(\alpha) = \prod_{\ell\in\{1,\ldots, T\}\setminus\{t\}}(1+\ell\alpha)$ (these are the polynomials obtained by clearing denominators of the hyperbolic discount factors). We argue that $\{Q_t(\alpha)\}_{t=1}^{T}$ are linearly independent over the vector space of polynomials in $\alpha$ of degree at most $T-1$. Indeed, if $Q_1(\alpha)f_1+\cdots+Q_T(\alpha)f_T = 0$, then we must have that $f_1 =\cdots = f_T = 0$ since at $\alpha=\frac{-1}{t}$ we get $Q_t(\alpha)f_t = 0$. Hence, $\mathcal{P}_{\mathcal{HD}}$ is simply $\mathcal{P}_{\mathcal{PW}}$ with $Q_t(\alpha) = \prod_{\ell\in\{1,\ldots, T\}\setminus\{t\}}(1+\ell\alpha)$.

\emph{Remark.} These bounds also hold in the scenario where the agent can report indifference in the data. More precisely, the condition for $x\succsim y$ is now a strict inequality, and we have three possible labels for the pair $(x, y)$: $+1$ indicates that $x\succsim y$, $-1$ indicates that $y\succsim x$, and $0$ indicates that $x\sim y$. Then, using a Hamiltonian path on $\{-1, 0, 1\}^n$, we can construct polynomials $P_1,\ldots, P_n$ such that $(\sgn(P_1(\delta)),\ldots, \sgn(P_n(\delta)))$ realizes all elements of $\{-1, 0, 1\}^n$ as $\delta$ ranges from $0$ to $1$ (where $\sgn$ is the true sign function).

\subsection{A Remark on Efficient Learnability}

While PAC learnability is a positive result, it does not take into account the computational complexity of computing a hypothesis. Blumer et. al. \cite{BLUMER} show that any learning rule that outputs a hypothesis consistent with the data seen yields with high probability a hypothesis that has very low error. However, if the problem of outputting a consistent hypothesis is computationally intractable, PAC learnability on its own is perhaps unsatisfying. In this section we note that the discounted utility models of intertemporal choice are \emph{efficiently} learnable. This is due to an algorithm of Grigor'ev and Vorobjov \cite{GRIGOR} for solving a system of polynomial inequalities.

For notational convenience, it will be useful to write $\mathcal{P} = \{\mathcal{P}^T\}_{T\ge 1}$, for each of the models above, where $\mathcal{P}^T$ is the collection of preference relations for a given $T$. Moreover, suppose acts are chosen from $[-1, 1]^T$ instead of $\RR^T$. It is clear that this does not change any of the analysis above.

Polynomial learnability, as defined by Blumer et. al. \cite{BLUMER}, stipulates that the learning rule be computable in $\text{poly}(1/\varepsilon, 1/\delta, T)$-time (where $\varepsilon$ and $\delta$ denote the error threshold and confidence threshold respectively). Polynomial learnability is equivalent to the task of outputting a hypothesis consistent with the given data set in polynomial time \cite{BLUMER}.

\begin{definition}
A \emph{randomized polynomial hypothesis finder (r-poly hy-fi)} for $\mathcal{P}$ is a randomized polynomial time algorithm that takes as input a sample of a preference relation in $\mathcal{P}$, and for some $\gamma > 0$, with probability at least $\gamma$ produces a hypothesis that is consistent with the sample.
\end{definition}

\begin{thm}
$\mathcal{P}$ is properly polynomially learnable if and only if there is an r-poly hy-fi for $\mathcal{P}$ and $VC(\mathcal{P}^T)$ grows only polynomially in $T$.
\end{thm}

We have just shown that $VC(\mathcal{P}^T)\sim\log(T-1)$. Using techniques involving algebraic geometry, Grigor'ev and Vorobjov \cite{GRIGOR} present an algorithm to solve a system of $N$ polynomial inequalities with a $\text{poly}(N, T)$ runtime. This serves as our r-poly hy-fi, and thus we obtain:

\begin{thm}
$\mathcal{P}_{\mathcal{ED}}$, $\mathcal{P}_{\mathcal{HD}}$, and $\mathcal{P}_{\mathcal{QHD}}$ are all properly polynomially learnable.
\end{thm}

\section{Active Learning}\label{active_learning}

In this section we study two models of active learning: \emph{stream-based selective sampling} and learning via \emph{membership queries}. We first define a distribution for which the disagreement coefficient of $\mathcal{P}_{\mathcal{ED}}$ is $2$, showing that disagreement methods (specifically the CAL algorithm \cite{COHN, DASGUPTA, HANNEKE}) in the stream-based model can yield an exponential improvement over the sample complexity of PAC learning (thus proving Theorem \ref{finitemu}).

We then consider learning via membership queries and show that in this setting even very naive algorithms outperform the disagreement methods in the stream-based model (that is, the analyst needs to ask fewer questions to the agent in order to learn his preference than the number of label requests he would need to make using disagreement methods). 

\subsection{Preliminaries}

For notational convenience, $\succsim_{\delta}$ will to refer to the preference relation in $\mathcal{P}_{\mathcal{ED}}$ with discounting factor $\delta$.

Let $\mu$ be a distribution on $\RR^T\times\RR^T$. $\mu$ induces a metric on $\mathcal{P}_{\mathcal{ED}}$ by $d(\succsim_{\delta}, \succsim_{\gamma}) = \mu(\succsim_{\delta}\triangle\succsim_{\gamma})$, and thus we can define the closed ball of radius $R$ centered at $\succsim_{\delta}$ by $$B(\succsim_{\delta}, R) = \{\succsim_{\gamma} : d(\succsim_{\delta}, \succsim_{\gamma})\le R\}.$$ For $V\subseteq\mathcal{P}_{\mathcal{EU}}$, the disagreement region of $V$, $\Dis(V)$ is defined by $$\Dis(V) = \{(x, y)\in\RR^T\times\RR^T : \exists\succsim_{\delta},\succsim_{\gamma}\in V\,\, s.t.\,\, (x, y)\in\,\succsim_{\delta}\triangle\succsim_{\gamma} \} = \displaystyle\bigcup_{\succsim_{\delta}, \succsim_{\gamma}\in V} (\succsim_{\delta}\triangle\succsim_{\gamma})$$ Intuitively, $\Dis(V)$ is the collection of points $(x, y)$ such that we can find two hypothesis relations in the current version space that rank $x$ and $y$ differently.

If $\succsim_{\delta}$ is the target preference relation, the disagreement coefficient of $\succsim_{\delta}$ with respect to $\mu$ is the quantity
$$\theta = \sup_{R > 0}\frac{\mu(\Dis(B(\succsim_{\delta}, R))}{R}$$

\subsection{Disagreement based active learning}

In this subsection, we define a distribution $\mu$ on $\RR^T\times\RR^T$ and show that the disagreement coefficient of $\mathcal{P}_{\mathcal{ED}}$ with respect to $\mu$ is $2$.

\subsubsection*{Choosing a measure}

The main challenge here is that $\theta$ depends on the underlying distribution over $\RR^T\times\RR^T$. Since preferences are polynomial inequalities, the disagreement coefficient seems to lend itself to a characterization involving polynomials and their roots, which is the motivation for our choice of distribution. For a general distribution $\mu$ over $\RR^T\times\RR^T$, it is not clear how to compute the disagreement coefficient. 

We show that $\theta = 2$ for a suitably chosen distribution on $\RR^T\times\RR^T$, which is induced by the Lebesgue measure on $(0, 1)^{T-1}$. This allows us to work with a measure on sets of roots of polynomials that arise from the definition of the preference relations.

Let $\mu^{**}$ be a measure on $(0, 1)^{T-1}$. We interchangeably represent elements of $\RR^{T}$ as polynomials $P$ of degree at most $T-1$ or as $T-1$-tuples of coefficients. Let $\sim$ be the equivalence relation on $\RR^T$ defined by $P\sim Q\iff P = cQ$ for some constant $c$, and let $\RR^T/\sim$ be the resulting quotient space. Let $g:(0, 1)^{T-1}\to\RR^T/\sim$ be the map taking a tuple of roots to the equivalence class of the polynomials with those roots, and let $h:\RR^T\times\RR^T\to\RR^T/\sim$ be the map $h(x, y) = [x-y]$. 

We define the following measures $\mu^*$ and $\mu$ on $g((0, 1)^{T-1})\subset\RR^T/\sim$ and $h^{-1}(g((0, 1)^{T-1}))\subset\RR^T\times\RR^T$ respectively. 

\begin{itemize}
	\item Define $\mu^*$ on all sets $S\subset g((0, 1)^{T-1})$ such that $$\{(r_1(P),\ldots, r_{T-1}(P))\in(0, 1)^{T-1} : P\in S\}$$ (where $r_1(P),\ldots r_{T-1}(P)$ denote the roots of $P$) is $\mu^{**}$-measurable, for which we set $$\mu^*(S) = \mu^{**}(\{(r_1(P),\ldots, r_{T-1}(P))\in(0, 1)^{T-1} : P\in S\}).$$ 
	\item Define $\mu$ on all sets $S\subset h^{-1}(g((0, 1)^{T-1}))$ such that $$\{[z]\in g((0, 1)^{T-1}) : \exists (x, y)\in S\, s.t. \, z \sim x-y\}$$ is $\mu^*$-measurable, for which we set $$\mu(S) = \mu^*(\{[z]\in g((0, 1)^{T-1}) : \exists (x, y)\in S\, s.t. \, z \sim x-y\}).$$
\end{itemize}

Intuitively, $\mu^*$ is defined only on those polynomials that have all their roots in $(0, 1)$. When $T=2$, this is a desirable property since the analyst is only presented with polynomials that have some disagreement in $(0, 1)$. He is not presented with meaningless polynomials that are, for example, always positive on $(0, 1)$ (the analyst has nothing to learn from such polynomials since such a polynomial will be preferred to $0$ for all $\delta\in (0, 1)$). For $T\ge 3$, this is a more restrictive property since the analyst is only presented with polynomials that have all $T-1$ roots in $(0, 1)$.


Let $\mu^{**}$ be the product Lebesgue measure on $(0, 1)^{T-1}$. Choosing $\mu^{**}$ in this fashion allows us to neatly characterize $B(\succsim_{\delta}, R)$.

Let $X_1,\ldots, X_{T-1}$ be uniform i.i.d. random variables on $(0, 1)$, and let $Y_{\delta, \gamma}$ be the random variable $Y_{\delta, \gamma} = |\{i : X_i\text{ is between }\delta\text{ and }\gamma\}|$. Let $E^{odd}_{\delta, \gamma}$ denote the event that $Y_{\delta, \gamma}$ is odd, let $E^k_{\delta, \gamma}$ denote the event $Y_{\delta, \gamma} = k$, and let $E^{\ge k}_{\delta, \gamma}$ denote the event $Y_{\delta, \gamma}\ge k$.

\begin{lem}\label{ball}$\succsim_{\gamma}\in B(\succsim_{\delta}, R)$ if and only if $\mathbb{P}[E^{odd}_{\delta, \gamma}]\le R$. \end{lem} 

\begin{proof}
Given $(x, y)\in\RR^T\times\RR^T$, let $P_{x-y}(X) = \sum_{t=1}^T X^{t-1}\cdot (x_t - y_t)$. Then, \begin{align*} \succsim_{\gamma}\in B(\succsim_{\delta}, R) &\iff \mu(\{(x, y)\in h^{-1}(g((0, 1)^{T-1})) : \sgn(P_{x-y}(\delta))\neq\sgn(P_{x-y}(\gamma)))\})\le R\\ &\iff \mu^*(\{[P]\in g((0, 1)^{T-1}) : \sgn(P(\delta))\neq\sgn(P(\gamma))\})\le R \\ &\iff \mu^{**}(\{(r_1,\ldots, r_{T-1})\in (0, 1)^{T-1} : \sgn(\textstyle\prod (\delta-r_i))\neq\sgn(\textstyle\prod (\gamma-r_i))\})\le R \end{align*}

But $\sgn(\textstyle\prod (\delta-r_i))\neq\sgn(\textstyle\prod (\gamma-r_i))$ occurs exactly when an odd number of roots lie between $\gamma$ and $\delta$ (modulo a set of measure $0$ since the probability that we have a root of multiplicity greater than $1$ is $0$). 
\end{proof}

While it seems difficult to write down a general characterization of $\mu$, Propositions \ref{mu*} and \ref{mu} give some basic observations regarding the $\sigma$-algebras on which $\mu^*$ and $\mu$ are defined. We defer their proofs (along with a description of $\mu^*$ in the case $T=2$) to the appendix:

\begin{prop}\label{mu*} The $\sigma$-algebra on $g((0, 1)^{T-1})$ is the Borel $\sigma$-algebra.\end{prop}

The $\sigma$-algebra induced on $h^{-1}(g((0, 1)^{T-1}))$ does not appear to yield a clean characterization, but we can show the weaker statement that $\mu$ is a Borel measure, i.e. it is defined on all open sets of $h^{-1}(g((0, 1)^{T-1}))$.

\begin{prop}\label{mu} $\mu$ is a Borel measure on $h^{-1}(g((0, 1)^{T-1}))$. \end{prop}


\subsubsection*{Computing $\theta$}

We now show $\theta = 2$ for the distribution $\mu$ as chosen above. We use the notation $d = d_{\delta, \gamma} := |\delta-\gamma|$ to denote the distance between $\delta$ and $\gamma$. Let $\succsim_{\delta}$ be the target preference relation.

First, note that since $Y_{\delta, \gamma}$ is distributed according to $\text{Bin}(T-1, d)$, $\PP[E^{odd}_{\delta, \gamma}] = \frac{1-(1-2d)^{T-1}}{2}$\footnote{This is due to the general fact that if $X$ is a random variable distributed according to $\text{Bin}(n, p)$, the probability that $X$ is odd is $\frac{1-(1-2p)^n}{2}$.}. We us this fact to derive an explicit description of the preference relations $\succsim_{\gamma}$ contained in the ball $B(\succsim_{\delta}, R)$ in terms of $d$. We break the analysis up into a few cases.

First, when $R > \frac{1}{2}$, we have $\sup_{R > \frac{1}{2}} \frac{\mu(\Dis(B(\succsim_{\delta}, R)))}{R} = 2$. Let $R\le\frac{1}{2}$. By Lemma \ref{ball}, \begin{equation}\succsim_{\gamma}\in B(\succsim_{\delta}, R)\iff \PP[E^{odd}_{\delta, \gamma}] = \frac{1-(1-2d)^{T-1}}{2}\le R\label{eq:ball}\end{equation} Suppose $d\le\frac{1}{2}$. Then $1-2d$ and $1-2R$ are both non-negative, so rearranging Equation \eqref{eq:ball} yields \begin{equation}d\le\frac{1-(1-2R)^{1/(T-1)}}{2}.\label{cond1}\end{equation} Suppose $d > \frac{1}{2}$, so $1-2d < 0$. Rearranging Equation \eqref{eq:ball}, we get $1-2R\le (1-2d)^{T-1}$. If $T-1$ is odd, $(1-2d)^{T-1}$ is negative, so $1-2R\le (1-2d)^{T-1}$ does not hold. Thus, when $T-1$ is odd the ball consists of $\succsim_{\gamma}$ such that $d$ satisfies condition \eqref{cond1}. If $T-1$ is even, $(1-2d)^{T-1}$ is positive, so we get \begin{equation}d\ge\frac{1+(1-2R)^{1/(T-1)}}{2}.\label{cond2}\end{equation} Thus, when $T-1$ is even the ball consists of $\succsim_{\gamma}$ such that $d$ satisfies conditions \eqref{cond1} or \eqref{cond2}.

Now, the disagreement region of $B(\succsim_{\delta}, R)$ consists of all points $(x, y)$ such that the polynomial $P_{x-y}$ (as defined in Lemma \ref{ball}) has a root $\gamma$ such that $\succsim_{\gamma}\in B(\succsim_{\delta}, R)$ (since we can find two hypotheses that disagree on $(x, y)$ by taking a point slightly below $\gamma$ and a point slightly above $\gamma$ such that $P_{x-y}$ has no sign changes in between). Hence, with $R_1 = \frac{1-(1-2R)^{1/(T-1)}}{2}$ and $R_2 = \frac{1+(1-2R)^{1/(T-1)}}{2}$, we have that $$\mu(\Dis(B(\succsim_{\delta}, R))) = \begin{cases}\PP[E_{\delta-R_1, \delta+R_1}^{\ge 1}] & \text{if } T-1\text{ is odd} \\ \PP[E_{\delta-R_1, \delta+R_1}^{\ge 1}\cup E_{0, \delta-R_2}^{\ge 1}\cup E_{\delta+R_2, 1}^{\ge 1}] & \text{if } T-1\text{ is even}\end{cases}$$

We have $$\PP[E_{\delta-R_1, \delta+R_1}^{\ge 1}] = 1-(1-2R_1)^{T-1} = 2R,$$ and $$\PP[E_{\delta-R_1, \delta+R_1}^{\ge 1}\cup E_{0, \delta-R_2}^{\ge 1}\cup E_{\delta+R_2, 1}^{\ge 1}] = 1 - (2(R_2 - R_1))^{T-1} = 1-2^{T-1}(1-2R).$$

Therefore, when $T-1$ is odd $\sup_{0 < R\le 1/2}\frac{\mu(\Dis(B(\succsim_{\delta}, R)))}{R} = 2$ and when $T-1$ is even $$\sup_{0 < R \le 1/2}\frac{\mu(\Dis(B(\succsim_{\delta}, R)))}{R} =\sup_{0 < R \le 1/2} \frac{1-2^{T-1}(1-2R)}{R} = 2,$$ which is achieved at $R = 1/2$ since $\frac{1-2^{T-1}(1-2R)}{R}$ is increasing on $0 < R\le\frac{1}{2}$.

Finally, $\theta = \sup_{R > 0}\frac{\mu(\Dis(B(\succsim_{\delta}, R)))}{R} = 2$.

We have thus established Theorem \ref{finitemu}: 

\begin{thmn}[\ref{finitemu}]
There exists a distribution $\mu$ on $\RR^T\times\RR^T$ for which the disagreement coefficient of $\mathcal{P}_{\mathcal{ED}}$ is $\theta = 2$. Thus, for this distribution, $$\ell_{CAL}(\varepsilon) = \widetilde{O}\left(\log T\log\frac{1}{\varepsilon}\right),$$ where the $\widetilde{O}$ notation suppresses terms that are logarithmic in $\log T$ and $\log 1/\varepsilon$.
\end{thmn}

\subsection{Learning via membership queries}\label{sec:memqueries}
In this subsection, we present a simple membership queries algorithm that outperforms the guarantees provided by the disagreement based CAL algorithm. For simplicity, we restrict attention to preference models that are parametrized by a single parameter (e.g. $\mathcal{P}_{\mathcal{ED}}$ and $\mathcal{P}_{\mathcal{HD}}$).

Let $g_1,\ldots, g_T : \RR\to\RR$ be a collection of functions such that there exist $1\le t_1, t_2\le T$ satisfying
\begin{enumerate}
\item $M:=\sup_{\delta} \frac{g_{t_1}(\delta)}{g_{t_2}(\delta)}$ is finite, and
\item The map $\delta\mapsto\frac{g_{t_1}(\delta)}{g_{t_2}(\delta)}$ satisfies an inverse Lipschitz condition with constant $C$ : $$|\delta - \delta'|\le C \left|\frac{g_{t_1}(\delta)}{g_{t_2}(\delta)}-\frac{g_{t_1}(\delta')}{g_{t_2}(\delta')}\right|.$$
\end{enumerate}

Consider the model of preference relations $\mathcal{P}$ parametrized by $\delta$ such that $$x\succsim y \,\, \text{ if and only if } \sum_{t=1}^T g_t(\delta)x_t\ge\sum_{t=1}^T g_t(\delta)y_t.$$

\begin{propn}[\ref{MembershipQueries}]  There exists an algorithm that takes as input $\varepsilon > 0$ and using $O(\log 1/\varepsilon)$ membership queries outputs $\delta^h$ such that $|\delta-\delta^h|\le\varepsilon$, where $\delta$ parametrizes the target preference relation in $\mathcal{P}$. \end{propn} 

\begin{proof} Fix a $\rho > 0$ and an $\eta$-cover of $[0, M\rho]$, where $0 < \eta\le\frac{\rho\varepsilon}{C}$. Let $b_{\rho}$ be the quantity such that the agent is indifferent between receiving a payoff of $\rho$ at time $t_1$ or receiving a payoff of $b_{\rho}$ at time $t_2$, i.e. $b_{\rho}$ solves $$g_{t_2}(\delta)b_{\rho} = g_{t_1}(\delta)\rho.$$ By running a binary search over the $\eta$-cover of $[\rho, M\rho]$, the analyst can find an approximation $b_{\rho}^h$ to the indifference point for which $|b_{\rho} - b_{\rho}^{h}|\le\eta$ (the binary search is performed on the parameter $b_{\rho}^h$ by requesting labels for pairs of the form $(\rho e_{t_1}, b_{\rho}^h e_{t_2})$). The analyst then outputs the $\delta^h$ that solves $g_{t_2}(\delta^h)b_{\rho}^h = g_{t_1}(\delta^h)\rho$.

We have $$|\delta - \delta^h|\le C\left|\frac{g_{t_1}(\delta)}{g_{t_2}(\delta)}-\frac{g_{t_1}(\delta^h)}{g_{t_2}(\delta^h)}\right| = C\left| \frac{b_{\rho}}{\rho} - \frac{b_{\rho}^h}{\rho}\right|\le \frac{C\eta}{\rho}\le\varepsilon,$$ as desired.

Since $M:=\sup_{\delta} \frac{g_{t_1}(\delta)}{g_{t_2}(\delta)}$ is finite, $b_{\rho}\le M\rho$, so the binary search over the $\eta$-cover of $[0, M\rho]$ terminates.
\end{proof} 

\emph{Remark.} Outputting a hypothesis parameter $\delta^h$ that is $\varepsilon$-close to $\delta$ is a reasonable measurement for the error of learning via membership queries since there is no underlying distribution providing points to the analyst. However, note that for a distribution on $\RR^T\times\RR^T$, a hypothesis close to the target parameter implies the set of misclassified points is assigned a small measure, due to continuity of measure.

The main feature of this algorithm is that its query complexity has no dependence on the number of time periods $T$. Both $\mathcal{P}_{\mathcal{ED}}$ and $\mathcal{P}_{\mathcal{HD}}$ fit the conditions of Proposition \ref{MembershipQueries}, and thus we obtain a large improvement over the guarantees provided by disagreement methods in the stream-based model. Such methods assume no extra knowledge about the problem domain and are written to fit a wide class of learning problems. When we are learning economic parameters, membership queries allow us to take advantage of the extra structure present in preference models. 

\subsection*{Acknowledgements}
We would like to thank Federico Echenique and Adam Wierman for several helpful comments and suggestions.

\appendix

\section{Properties of $\mu$ and $\mu^*$}

\begin{proof}[Proof of Proposition \ref{mu*}]

Let $\mathcal{B}$ denote the Borel $\sigma$-algebra on $(0, 1)^{T-1}$.

Let $X = \mathbb{R}^T/\sim$ endowed with the quotient topology, and $g^*: [0,1]^{T-1} \to X$ be the map $g^*(r_1,\ldots, r_{T-1}) = [\prod (x-r_i)]$. Explicitly, the terms of $g^*(r_1,\ldots, r_{T-1})$ are given by symmetric sums: $$g^*(r_1,\ldots, r_{T-1}) = \left(c, -c\sum_{i} r_i, c\sum_{i, j} r_ir_j, \ldots,  (-1)^{T-1}cr_1\cdots r_{T-1}\right),$$ where $c$ is the appropriate constant for the representative of the equivalence class. Each symmetric sum is a continuous function of $T-1$ variables, so $g^*$ is continuous. Moreover, note that $g^*$ is injective. Then, with $Y = g^*([0, 1]^{T-1})$, we have that $g^*:[0, 1]^{T-1}\to Y$ is a continuous bijection from a compact set into a Hausdorff space. Hence, $g^*$ is a homeomorphism. Then $g$, which is the restriction of $g^*$ to $(0, 1)^{T-1}$ is a homeomorphism onto $Z:= g((0, 1)^{T-1})$. Thus, the $\sigma$-algebra $g(\mathcal{B})$ that we obtain on $Z$ is the Borel $\sigma$-algebra.
\end{proof}

When $T = 2$, we can give an explicit description of $\mu^*$. Identify $\RR^2/\sim$ with the unit circle. Then, a degree $1$ polynomial $P$ is identified with the point $(\cos\theta, \sin \theta)$, where $P(x) = (\cos\theta)x + \sin\theta$. $Z:= g((0, 1)^{T-1})$ consists of the boundary of the unit circle for which the argument $\theta$ satisfies $-\tan\theta\in (0, 1)$. This is satisfied precisely for $\theta\in (3\pi/4, \pi)\cup (7\pi/4, 2\pi)$. Hence, if $U$ is a basic open subset of $\{(\cos\theta, \sin\theta) : \theta\in(3\pi/4, \pi)\cup (7\pi/4, 2\pi)\}$, we can write $U = \{(\cos\theta,\sin\theta) : \theta_1 < \theta < \theta_2\}$ with $\theta_1,\theta_2$ both in the same segment of the unit circle and $$\mu^*(U) = \mu^{**}(\{-\tan\theta : \theta_1 < \theta < \theta_2\}) = |\tan\theta_1 - \tan\theta_2|.$$

\begin{proof}[Proof of Proposition \ref{mu}]
Let $h:\RR^T\times\RR^T\to\RR^T/\sim$ be the map $h(x, y) = [x-y]$. Let $V\subseteq h^{-1}(g((0, 1)^{T-1}))$ be open. We show that $$ h(V) = \{[z] : \exists (x, y)\in V (z = x-y)\}$$ is open. Indeed, let $z = x - y$ for $(x, y)\in V$ and choose $\varepsilon$ small enough such that the square with vertices $\{(x+\varepsilon, y+\varepsilon), (x+\varepsilon, y-\varepsilon), (x-\varepsilon, y+\varepsilon), (x-\varepsilon, y-\varepsilon)\}$ is contained in $V$. Then, for any $\lambda\le\varepsilon$, $[z+\lambda] = [(x+\lambda) - y]$ with $(x+\lambda, y)\in V$ and $[z-\lambda] = [x - (y+\lambda)]$ with $(x, y+\lambda)\in V$, so in particular the open ball with radius $\lambda$ centered at $[z]$ is contained in $h(V)$.
\end{proof}


\end{document}